\documentclass[runningheads]{llncs}
\usepackage{graphicx}
\usepackage{amssymb}
\usepackage{amsmath}
\newtheorem{observation}[theorem]{Observation}

 \pdfoutput=1
 
\begin{document}
\title{The Fault-Tolerant Metric Dimension of Cographs}
%
%
\author{Duygu Vietz\inst{1}\orcidID{0000-0001-6881-7832} \and
Egon Wanke\inst{1}}
\authorrunning{D. Vietz et al.}
%
\institute{Heinrich-Heine-University Duesseldorf, Universitaetsstr. 1, 40225 Duesseldorf, Germany }  

\maketitle              
\begin{abstract}
A vertex set $U \subseteq V$ of an undirected graph $G=(V,E)$ is a \textit{resolving set} for $G$ if for every two distinct vertices $u,v \in V$ there is a vertex $w \in U$ such that the distance between $u$ and $w$ and the distance between $v$ and $w$ are different. A resolving set $U$ is {\em fault-tolerant} if for every vertex $u\in U$ set $U\setminus \{u\}$ is still a resolving set. {The \em (fault-tolerant) Metric Dimension} of $G$ is the size of a smallest (fault-tolerant) resolving set for $G$. The {\em weighted (fault-tolerant) Metric Dimension} for a given cost function $c: V \longrightarrow \mathbb{R}_+$ is the minimum weight of all (fault-tolerant) resolving sets. Deciding whether a given graph $G$ has (fault-tolerant) Metric Dimension at most $k$ for some integer $k$ is known to be NP-complete. The weighted fault-tolerant Metric Dimension problem has not been studied extensively so far. In this paper we show that the weighted fault-tolerant metric dimension problem can be solved in linear time on cographs. 

\end{abstract}
\section{Introduction}  
\label{secIntro}  

An undirected graph $G=(V,E)$ has {\em metric dimension at most $k$} if there is a vertex set $U \subseteq V$ such that $|U| \leq k$ and $\forall u,v \in V$, $u \not=v$, there is a vertex $w \in U$ such that $d_G(w,u) \not= d_G(w,v)$, where $d_G(u,v)$ is the distance (the length of a shortest path in an unweighted graph) between $u$ and $v$. We call $U$ a {\em resolving set}. Graph $G$ has {\em fault-tolerant metric dimension at most $k$} if for a resolving set $U$ with $|U|\leq k$ it holds that for every $u\in U$ set $U \setminus \{u\}$ is a resolving set for $G$.
The metric dimension of $G$ is the smallest integer $k$ such that $G$ has metric dimension at most $k$ and the fault-tolerant metric dimension of $G$ is the smallest integer $k$ such that $G$ has fault-tolerant metric dimension at most $k$. The metric dimension was independently introduced by Harary, Melter \cite{HM76} and Slater \cite{Sla75}.

If for three vertices $u,v\in V$, $w\in U$, we have $d_G(w,u) \not= d_G(w,v)$, then we say that $u$ and $v$ are {\em resolved} by vertex $w$. The {\em metric dimension} of $G$ is the size of a minimum resolving set and the {\em fault-tolerant metric dimension} is the size of a minimum fault-tolerant resolving set. In certain applications, the vertices of a (fault-tolerant) resolving set are also called {\em resolving vertices}, {\em landmark nodes} or {\em anchor nodes}. This is a common naming particularly in the theory of sensor networks. 

Determining the metric dimension of a graph is a problem that has an impact on multiple research fields such as chemistry \cite{CEJO00}, robotics \cite{KRR96}, combinatorial optimization \cite{ST04} and sensor networks \cite{HW12}. Deciding whether a given graph $G$ has metric dimension at most $k$ for a given integer $k$ is known to be NP-complete for general graphs \cite{GJ79}, planar graphs \cite{DPSL12}, even for those with maximum degree 6 and Gabriel unit disk graphs \cite{HW12}. Epstein et al.\ showed the NP-completeness for split graphs, bipartite graphs, co-bipartite graphs and line graphs of bipartite graphs \cite{epstein2015weighted} and Foucaud et al.\ for permutation and interval graphs \cite{foucaud2015algorithms}\cite{foucaud2017identification}.

There are several algorithms for computing the metric dimension in polynomial time for special classes of graphs, as for example for trees \cite{CEJO00,KRR96}, wheels \cite{HMPSCP05}, grid graphs \cite{MT84}, $k$-regular bipartite graphs \cite{SBSSB11}, amalgamation of cycles \cite{IBSS10}, outerplanar graphs \cite{DPSL12}, cactus block graphs \cite{hoffmann2016linear}, chain graphs \cite{fernau2015computing}, graphs with a bounded number of resolving vertices in every EBC \cite{VHW19}. The approximability of the metric dimension has been studied for bounded degree, dense, and general graphs in \cite{HSV12}. Upper and lower bounds on the metric dimension are considered in \cite{CGH08,CPZ00} for further classes of graphs.

There are many variants of the Metric Dimension problem.  The weighted version was introduced by Epstein et al.\ in \cite{epstein2015weighted}, where they gave a polynomial-time algorithms on paths, trees and cographs. Hernando et al.\ investigated the fault-tolerant Metric Dimension in \cite{hernando2008fault}, Estrada-Moreno et al.\ the $k$-metric Dimension in \cite{estrada2013k} and Oellermann et al.\ the strong metric Dimension in \cite{oellermann2007strong}.
 
The parameterized complexity was investigated by Hartung and Nichterlein. They showed that for the standard parameter the problem is $W[2]$-complete on general graphs, even for those with maximum degree at most three \cite{hartung2013parameterized}. Foucaud et al.\ showed that for interval graphs the problem is FPT for the standard parameter \cite{foucaud2015algorithms}\cite{foucaud2017identification}. Afterwards Belmonte et al.\ extended this result to the class of graphs with bounded treelength, which is a superclass of interval graphs and also includes chordal, permutation and AT-free graphs \cite{belmonte2017metric}. 
  
\bigskip 
In this paper we show that the weighted fault-tolerant metric dimension problem can be solved in linear time on cographs and give an algorithm that computes a minimum weight fault-tolerant resolving set.

\section{Definitions and Basic Terminology}
\label{secDef}

We consider {\em graphs} $G=(V,E)$, where $V$ is the set of {\em vertices} and $E$ is the set of {\ edges}. We distinguish between {\em undirected graphs} with edge sets $E \subseteq \{\{u,v\}~|~u,v \in V,~u \not=v\}$ and {\em directed graphs} with edge sets $E \subseteq V \times V.$ Graph $G'=(V',E')$ is a {\em subgraph} of $G=(V,E)$, if $V' \subseteq V$ and $E' \subseteq E$. It is an {\em induced subgraph} of $G$, denoted by $G|_{V'}$, if $E' = E \cap \{\{u,v\}~|~ u,v \in V'\}$ or $E' = E \cap (V'\times V')$, respectively. Vertex $u \in V$ is called a {\em neighbour} of vertex $v \in V$, if $\{u,v\} \in E$ in an undirected graph or $(u,v)\in E$ ($(v,u) \in E$) in a directed graph. With $N(u) = \{v \ | \ \{u,v\} \in E \}$ we denote the {\em open neighbourhood} of a vertex $u$ in an undirected graph and with $N[u] = N(u) \cup \{u\}$ we denote the {\em closed neighbourhood} of a vertex $u$. 

A sequence of $k+1$ vertices $(u_1,\ldots,u_{k+1})$, $k \geq 0$, $u_i \in V$ for $i=1, \ldots , k+1$, is an {\em undirected path of length $k$}, if $\{u_i,u_{i+1}\} \in E$ for $i=1,\ldots,k$. The vertices $u_1$ and $u_{k+1}$ are the {\em end vertices} of undirected path $p$. The sequence $(u_1,\ldots,u_{k+1})$ is a {\em directed path of length $k$}, if $(u_i,u_{i+1}) \in E$ for $i=1,\ldots,k$. Vertex $u_1$ is the start vertex and vertex $u_{k+1}$ is the end vertex of the directed path $p$. A path $p$ is a {\em simple path} if all vertices are mutually distinct.

An undirected graph $G$ is {\em connected} if there is a path between every pair of vertices. An undirected graph $G$ is {\em disconnected} if it is not connected. A {\em connected component} of an undirected graph $G$ is a connected induced subgraph $G'=(V',E')$ of $G$ such that there is no connected induced subgraph $G''=(V'',E'')$ of $G$ with $V'\subseteq V''$ and $|V'|<|V''|$. A vertex $u \in V$ is a {\em separation vertex} of an undirected graph  $G$ if $G|_{V \setminus \{u\}}$ (the subgraph of $G$ induced by $V \setminus \{ u \} $) has more connected components than $G$. Two paths $p_1=(u_1,\ldots,u_k)$ and $p_2=(v_1,\ldots,v_l)$ are {\em vertex-disjoint} if $\{u_2,\ldots,u_{k-1}\} \cap \{v_2\ldots,v_{l-1}\} = \emptyset$. A graph $G=(V,E)$ with at least three vertices is {\em biconnected}, if for every vertex pair $u,v \in V$, $u \not= v$, there are at least two vertex-disjoint paths between $u$ and $v$. A {\em biconnected component} $G'=(V',E')$ of $G$ is an induced biconnected subgraph of $G$ such that there is no biconnected induced subgraph $G''=(V'',E'')$ of $G$ with $V' \subseteq V''$ and $|V'|<|V''|$.
The {\em distance} $d_G(u,v)$ between two vertices $u,v$ in a connected undirected graph $G$ is the smallest integer $k$ such that there is a path of length $k$ between $u$ and $v$. The {\em distance} $d_G(u,v)$ between two vertices $u,v$ such that there is no path between $u$ and $v$ in $G$ is $\infty$. The {\em complement} of an undirected graph $G=(V,E)$ is the graph $\bar G = (V, \{ \{u,v\}  \; | \; u,v \in V, \{u,v\} \notin E \})$.


\begin{definition}[Cograph]
An undirected Graph $G$ is a {\em cograph}, if
\begin{itemize}
 \item $G=(\{u\}, \emptyset)$ or
 \item $G=(V_1 \cup V_2, E_1 \cup E_2)$ for two cographs $G_1=(V_1,E_1)$ and $G_2=(V_2,E_2)$ or
 \item $G= \overline{H}$ for a cograph $H$.
\end{itemize}
\end{definition}

A cograph contains no induced $P_4$, therefore the diameter of a connected cograph $G$ is at most 2. That is, the distance between two arbitrary verices $u,v$ in $G$ is either 0 or 1 or 2.

\begin{definition}[Resolving set, metric dimension]
Let $G=(V,E)$ be an undirected graph and let $c:V \longrightarrow \mathbb{R}_+$ be a function that assigns to every vertex a non-negative weight. A vertex set $R \subseteq V$ is a {\em resolving set} for $G$ if for every vertex pair $u,v \in V,$ $u \not = v$, there is a vertex $w \in R$ such that $d_G(u,w) \not= d_G(v,w)$. A resolving set $R\subseteq V$ has weight $k \in \mathbb{N}$, if $\sum_{v \in R} c(v)=k$.
The set $R$ is a {\em minimum resolving set for $G$}, if there is no resolving set $R'\subseteq V$ for $G$ with $|R'| < |R|$. The set $R$ is a {\em minimum weight resolving set for $G$}, if there is no resolving set $R'\subseteq V$ for $G$ with $\sum_{v \in R'} c(v) < \sum_{v \in R} c(v)$.  An undirected graph $G=(V,E)$ has {\em metric dimension} $k\in\mathbb{N}$, if $k$ is the smallest positive integer such that there is a resolving set for $G$ of size $k$. An undirected graph $G=(V,E)$ has {\em weighted metric dimension} $k\in\mathbb{N}$ if $k$ is the smallest positive integer such that there is a resolving set for $G$ of weight $k$.
\end{definition}

\begin{definition}[Fault-tolerant resolving set, fault-tolerant metric dimension]
Let $G=(V,E)$ be an undirected graph and let $c:V \longrightarrow \mathbb{R}_+$ be a function that assigns to every vertex a non-negative weight. A vertex set $R \subseteq V$ is a {\em fault-tolerant resolving set} for $G$ if for an arbitrary vertex $r \in R$ set $R\setminus \{r\}$ is a resolving set. A fault-tolerant resolving set $R\subseteq V$ has weight $k \in \mathbb{N}$, if $\sum_{v \in R} c(v)=k$.
The set $R$ is a {\em minimum fault-tolerant resolving set for $G$}, if there is no fault-tolerant resolving set $R'\subseteq V$ for $G$ with $|R'| < |R|$. The set $R$ is a {\em minimum weight fault-tolerant resolving set for $G$}, if there is no fault-tolerant resolving set $R'\subseteq V$ for $G$ with $\sum_{v \in R'} c(v) < \sum_{v \in R} c(v)$. An undirected graph $G=(V,E)$ has {\em fault-tolerant metric dimension} $k\in\mathbb{N}$, if $k$ is the smallest positive integer such that there is a fault-tolerant resolving set for $G$ of size $k$. An undirected graph $G=(V,E)$ has {\em weighted fault-tolerant metric dimension} $k\in\mathbb{N}$, if $k$ is the smallest positive integer such that there is a fault-tolerant 
resolving set for $G$ of weight $k$.
\end{definition}

Equivalent to this definition one can say that a vertex set is a fault-tolerant resolving set if for every vertex pair there are two resolving vertices. Obviously every fault-tolerant resolving set is also a resolving set.

The concept of fault-tolerance can be extended easily on an arbitrary number of vertices, what is called the $k$-metric dimension in \cite{estrada2013k}, $k \in \mathbb{N}$. The $k$-metric dimension is the size of a smallest $k$-resolving set. A $k$-resolving set resolves every pair of vertices at least $k$ times. For $k=1$ a $k$-resolving set is a resolving set and for $k=2$ a $k$-resolving set is a fault-tolerant resolving set. One should note that for all $k>2$ there are graphs that does not have a $k$-resolving set (for example graphs with twin vertices), whereas for $k \leq 2$ the entire vertex set is a $k$-resolving set.


\begin{definition}
 Let $G=(V,E)$ be an undirected graph and $u,v \in V$, $u \neq v$. For two vertices $u,v \in V$ we call $N(u) \triangle N(v) = (N(u) \cup N(v) ) \setminus ( N(u) \cap N(v))$  the {\em symmetric difference} of $u$ and $v$. For a set $R\subseteq V$, we define the function $$h_R: V \times V \longrightarrow \mathbb{N}, \hspace{1cm} h_R(u,v)=  |(N(u) \triangle N(v) \cup \{u,v\} ) \cap R |$$ 
\end{definition}

$h_R(u,v)$ is the number of vertices in $R$ that are $u$ or $v$ or a neighbour of $u$, but not of $v$ or a neighbour of $v$, but not of $u$.

\begin{definition}[neighbourhood-resolving]
\label{NR}

Let $G=(V,E)$ be an undirected graph and
$u,v \in V$, $u \neq v$, and $R\subseteq V$. Set $R$ is called {\em neighbourhood-resolving} for $G$, if for every pair $u,v \in V$, $u \neq v$, we have $h_R(u,v)\geq 1$.
\end{definition}


A set $R$ is neighbourhood-resolving for $G$, if for every two vertices $u,v \notin R$ there is a vertex $w \in R$ that is neighbour of exactly one of the vertices $u$ and $v$. If $u \in R$ or $v\in R$ the value $h_R(u,v)$ is always at least 1. Obviously, every set that is neighbourhood-resolving for $G$ is also a resolving set for $G$.

\begin{definition}[2-neighbourhood-resolving]
\label{2NR}
Let $G=(V,E)$ be an undirected graph and $u,v \in V$, $u \neq v$, and $R\subseteq V$. Set $R$ is called {\em 2-neighbourhood-resolving} for $G$ if for every pair $u,v \in V$, $u \neq v$, we have $h_R(u,v)\geq 2$.
\end{definition}

A set $R$ is 2-neighbourhood-resolving for $G$ if 
\begin{itemize}
 \item for two vertices $u,v \in V\setminus R$ there are at least two vertices in $R$ that are neighbour of exactly one of the vertices $u$ and $v$ and
 \item for two vertices $u,v$ such that $u\in R$ and $v \notin R$ there is at least one vertex in $R$ that is neighbour of exactly one of the vertices $u$ and $v$.
 \end{itemize}

For $u,v \in R$ the value $h_R(u,v)$ is always at least two. Obviously, every set that is 2-neighbourhood-resolving for $G$ is also a fault-tolerant resolving set for $G$.

\begin{lemma}
\label{ft2NR}
Let $G=(V,E)$ be a connected cograph and $R \subseteq V$. Vertex set $R$ is a fault-tolerant resolving set for $G$ if and only if $R$ is 2-neighbourhood-resolving for $G$.
\end{lemma}

\begin{proof}
"$\Rightarrow$": Assume that $R$ is a fault-tolerant resolving set for $G$. We have to show that $R$ is 2-neighbourhood-resolving for $G$, so let $u,v \in V$ and $r_1, r_2 \in R$ be the vertices that resolve $u$ and $v$. 
\begin{enumerate}
 \item If $u,v \in R$, then obviously $h_R(u,v)\geq2$.
 \item If $u \in R$ and $v \notin R$, then either $d_G(u,r_1) \neq 0$ or $d_G(u,r_2) \neq 0$. Without loss of generality let $d_G(u,r_1) \neq 0$. Vertex $v \notin R$, so $d_G(v,r_1) \neq 0$. Since vertex $r_1$ resolves $u,v$ and $G$ is a connected cograph (and therefore the diameter is at most 2), $r_1$ has to be adjacent to exactly one of the vertices $u,v$. Thus, $r_1 \in u\triangle v \cap R $ and $u \in \{u,v\} \cap R$ and therefore $h_R(u,v)\geq 2$.
 \item If $u,v \notin R$, then the distance between $u$ and any vertex in $R$ and the distance between $v$ and any vertex in $R$ is not 0. Since $r_1$ and $r_2$ resolve $u$ and $v$ both are adjacent to exactly one of the vertices $u$ and $v$. Thus $r_1,r_2 \in N(u) \triangle N(v)$ and therefore $h_R(u,v) \geq 2$.
\end{enumerate}

 "$\Leftarrow$": Assume that $R$ is 2-neighbourhood-resolving for $G$. We have to show that $R$ is a fault-tolerant resolving set for $G$. We do this by giving two resolving vertices for every vertex pair $u,v \in V$.
 \begin{enumerate}
  \item If $u,v \in R$, there are obviously two vertices in $R$, which resolve $u$ and $v$.
  \item If $u \in R$ and $v \in V \setminus R$, then $u$ resolves $u,v$. Since $h_R(u,v) \geq 2$ and $|\{u,v\} \cap R|=1$, we have $|N(u) \triangle N(v) \cap R | \geq 1$. Thus, there is a vertex $r \in R$, that is adjacent to exactly one of the vertices $u,v$, so $r$ resolves $u,v$.
  \item If $u,v \in V \setminus R$, then $|\{u,v\} \cap R | =0$. Since $h_R(u,v)\geq 2$, it follows $|N(u) \triangle N(v) \cap R| \geq 2$. Thus, there are two vertices $r_1,r_2 \in R$, that are both adjacent to exactly one of the vertices $u,v$ and so $r_1,r_2$ resolve $u,v$.
 \end{enumerate}
\end{proof}

Note that this equivalence does not apply to disconnected cographs, see Figure \ref{Fig1}.

\begin{figure}
 \centering
 \includegraphics[width = 0.9 \textwidth]{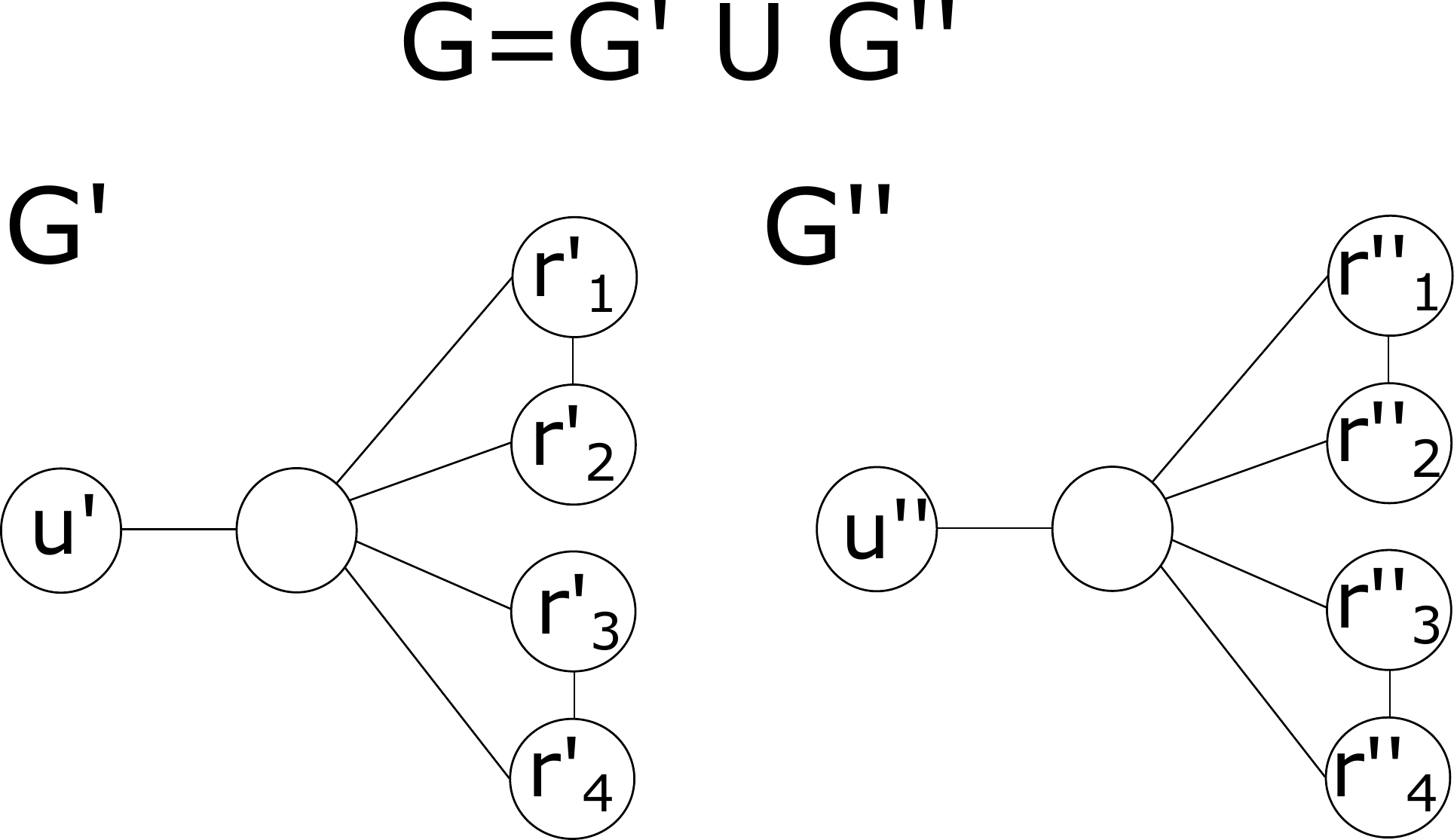}
 \caption{\label{Fig1} The figure shows the disconnected cograph $G=G' \cup G''$, build by the union of the two connected cographs $G'$ and $G''$. Let $R=R' \cup R''$ with $R'= \{r'_1, \ldots , r'_4\}$ and $R''= \{r''_1, \ldots , r''_4\}$. $R'$ is 2-neighbourhood-resolving and a fault-tolerant resolving set for $G'$ and $R''$ is 2-neighbourhood-resolving and a fault-tolerant resolving set for $G''$. $R$ is a fault-tolerant resolving set, but not 2NR for $G$, since $h_R(u',u'')=0$. $R$ is not a fault-tolerant resolving set for $\bar G$, since $u'$ and $u''$ are neighbour of every resolving vertex in $R$ in graph $\bar G$ and therefore cannot be resolved.}
\end{figure} 

Thus, we state that 2-neighbourhood-resolving implies fault-tolerance in a cograph, fault-tolerance implies  2-neighbourhood-resolving in a connected cograph, but not in a disconnected cograph.

\begin{lemma}
\label{comp1}
 Let $G=(V,E)$ be a cograph and $R \subseteq V$. If $R$ is 2-neighbourhood-resolving for $G$, then $R$ is also 2-neighbourhood-resolving for $\bar G$.
\end{lemma}

\begin{proof}
 Let $R \subseteq V$ be 2-neighbourhood-resolving for $G$, i.e. for $u,v \in V$ we have $h_R(u,v) =  |(N(u) \triangle N(v) \cup \{u,v\} ) \cap R | \geq 2$. We distinguish between the following cases:
\begin{enumerate}
 \item $u,v \in (N(u) \triangle N(v) \cup \{u,v\} ) \cap R$:
 Obviously, $u,v \in (N(u) \triangle N(v) \cup \{u,v\} ) \cap R$ in graph $\bar G$ and so $h_R(u,v) \geq 2$ in $\bar G$.
 
 \item $u \in (N(u) \triangle N(v) \cup \{u,v\} ) \cap R$ and $v \notin (N(u) \triangle N(v) \cup \{u,v\} ) \cap R$:
 Since $h_R(u,v) \geq 2$ there has to be a vertex $w \in N(u) \triangle N(v) \cap R$, what implies that $w$ is neighbour of either $u$ or $v$. Without loss of generality let $w$ be a neighbour of $u$. In graph $\bar G$ vertex $w$ is not a neighbour of $u$, but a neighbour of $v$. So, we still have two vertices $u, w \in  (N(u) \triangle N(v) \cup \{u,v\} ) \cap R$ in graph $\bar G$.
 
 \item $u,v \notin (N(u) \triangle N(v) \cup \{u,v\} ) \cap R$: Since $h_R(u,v) \geq 2$ there has to be two vertices $w_1, w_2 \in N(u) \triangle N(v) \cap R$, what implies that both are neighbour of exactly one of the vertices $u,v$. Therefore in graph $\bar G$ they are also neighbour of exactly one of the vertices $u,v$. So, we still have two vertices $w_1, w_2 \in  (N(u) \triangle N(v) \cup \{u,v\} ) \cap R$ in graph $\bar G$.
\end{enumerate}
\end{proof}

 Since 2-neighbourhood-resolving is equivalent to fault-tolerance in connected cographs, we get the following observation:
 
 \begin{observation}
  Let $G=(V,E)$ be a connected cograph and $R \subseteq V$. If $R$ is a fault-tolerant resolving set for $G$, then $R$ is also a fault-tolerant resolving set for the disconnected cograph $\bar G$.
 \end{observation}

 Note that a fault-tolerant resolving set $R$ for a disconnected cograph $G$ is not necessarily a fault-tolerant resolving set for $\bar G$, see Figure \ref{Fig1}.

\begin{lemma}
\label{ftunion}
 Let $G'=(V',E')$ and $G''=(V'',E'')$ be two connected cographs and $G=(V,E) $ with $V=V'\cup V''$ and $E=E' \cup E''$ be the disjoint union of $G'$ and $G''$. Let $R'$ be a fault-tolerant resolving set for $G'$ and $R''$ be a fault-tolerant resolving set for $G''$. Then $R=R' \cup R''$ is a fault-tolerant resolving set for $G$.
\end{lemma}

\begin{proof}
 We show that every pair $u,v \in V$ is resolved by two vertices in $R$. If $u,v \in V_1$ or $u,v \in V_2$ the pair is obviously resolved twice by vertices in $R_1 \subseteq R$ or $R_2 \subseteq R$. If $u \in V_1$ and $v \in V_2$ the pair is resolved by any two resolving vertices $r_1,r_2 \in R$, since either $u$ or $v$ will have distance $\infty$ to $r_1$ and $r_2$.   
\end{proof}

Note that $R$ is not necessarily 2-neighbourhood-resolving for $G$ (see Figure \ref{Fig1}).

\begin{definition}
 Let $G=(V,E)$ be a cograph and $R\subseteq V$ a fault-tolerant resolving set for $G$. A vertex $v \in V$ is called a {\em $k$-vertex} with respect to $R$, $k \in \mathbb{N}$, if $|N[v] \cap R| = k $.
\end{definition} 

A vertex $v \in V$ is a $k$-vertex, if it has $k$ vertices in its closed neighbourhood that are in $R$.


\begin{lemma}
\label{union}
  Let $G'=(V',E')$ and $G''=(V'',E'')$ be two connected cographs and $G=(V,E) $ with $V=V'\cup V''$ and $E=E' \cup E''$ be the disjoint union of $G'$ and $G''$. Let $R'$ be 2-neighbourhood-resolving for $G'$ and $R''$ be 2-neighbourhood-resolving for $G''$. Vertex set $R=R' \cup R''$ is 2-neighbourhood-resolving for $G$ if and only if
  \begin{enumerate}
   \item there is at most one $0$-vertex $v \in V$ with respect to $R$, i.e. there is no $0$-vertex $v \in V'$ with respect to $R'$ or there is no $0$-vertex $v \in V''$ with respect to $R''$ and
   \item there is no $0$-vertex $v \in V'$ with respect to $R'$, if there is a $1$-vertex in $V''$ with respect to $R''$ and
   \item there is no $1$-vertex in $V'$ with respect to $R'$, if there is a $0$-vertex in $V''$ with respect to $R''$.
  \end{enumerate}
\end{lemma}

\begin{proof}
"$\Rightarrow$": Assume that $R$ is 2-neighbourhood-resolving for $G$.
\begin{enumerate}
 \item We show that there is at most one 0-vertex in $V$ with respect to $R$. Assume there are two 0-vertices $u,v \in V$ with respect to $R$, i.e. $|N[u] \cap R| = 0$ and $|N[v] \cap R| = 0$. Then we have $h_R(u,v) = 0$, what contradicts the assumption that $R$ is 2-neighbourhood-resolving.
 \item We show that there is no 0-vertex in $V'$ with respect to $R'$ if there is a $1$-vertex in $V''$ with respect to $R''$. Assume that there is a 0-vertex in $u \in V'$ with respect to $R'$ and a $1$-vertex in $v \in V''$ with respect to $R''$. Then we have $h_R(u,v)=1$, what contradicts the assumption that $R$ is 2-neighbourhood-resolving.
 \item analogous to 2. 
 \end{enumerate}

 "$\Leftarrow$": Assume that the conditions 1., 2. and 3. hold. We  show that $R$ is 2-neighbourhood-resolving for $G$, i.e. for $u,v \in V$ we have $h_R(u,v) \geq 2$. For $u,v \in V'$ we have $h_{R'}(u,v)\geq 2$ and therefore also $h_R(u,v) \geq 2$. The same holds for $u,v \in V''$. Now let $u \in V'$ and $v \in V''$. $h_R(u,v)<2$ if and only if $|N[u] \cap R| + |N[v] \cap R| <2$, i.e. if 
 \begin{enumerate}
  \item $|N[u] \cap R| = 0$ and $|N[v] \cap R| = 0$ or
  \item $|N[u] \cap R| = 0$ and $|N[v] \cap R| = 1$ or
  \item $|N[u] \cap R| = 1$ and $|N[v] \cap R| = 0$
 \end{enumerate}
Conditions 1. - 3. guarantee that none of these three cases appear. 
\end{proof}

\begin{theorem}
 Let $G=(V,E)$ be a cograph. The weighted fault-tolerant metric dimension of $G$ can be computed in linear time.
\end{theorem}

\begin{proof}
 We describe a linear time algorithm for computing the weighted fault-tolerant metric dimension of a connected cograph. For disconnected cographs we apply the algorithm for every connected component with at least two vertices. If there are isolated vertices, then each of them has to be in every weighted fault-tolerant resolving set, except for the case that there is exactly one isolated vertex. To get the weighted fault-tolerant metric dimension of the disconnected input graph, we build the sum of the weights of all isolated vertices if there are at least two, and the weighted fault-tolerant metric dimension for each connected component with at least two vertices.
 
 To compute the weighted fault-tolerant metric dimension of a connected cograph $G=(V,E)$ it suffices to compute a set that is 2-neighbourhood-resolving for $G$ and has minimal costs, since fault-tolerant resolving and 2-neighbourhood-resolving sets are equivalent in connected cographs (Lemma \ref{ft2NR}). To compute a 2-neighbourhood-resolving set of minimum weight we use dynamic programming along the cotree $T=(V_T,E_T)$. The cotree $T$ of $G$ is a tree that describes the union and complementation of cographs. The inner nodes are either complementation-nodes or union-nodes. Every complementation-node has exactly one child and every union-node has exactly two children. The leafs of $T$ are the vertices of $G$.
  
 For every inner node of $T$ we compute bottom up different types of minimum weight 2-neighbourhood-resolving sets for the corresponding subgraph of $G$. First we compute the 2-neighbourhood-resolving sets for the fathers of the leafs. For every other inner node $v \in V_T$ we compute the 2-neighbourhood-resolving sets from the 2-neighbourhood-resolving sets of all children of $v$. Finally the minimum weight of all 2-neighbourhood-resolving sets at root $r$ of $T$ will be the minimum weight fault-tolerant metric dimension of $G$. From Lemma \ref{comp1} we know that, if a set is 2-neighbourhood-resolving for a cograph $G'$ then it is also 2-neighbourhood-resolving for $\bar G'$. The union of two fault-tolerant resolving sets is also a fault-tolerant resolving set (Lemma $\ref{ftunion}$), but the union of two 2-neighbourhood-resolving sets is not necessarily a 2-neighbourhood-resolving set. We have to guarantee that the union of two 2-neighbourhood-resolving sets is also 2-neighbourhood-resolving, according to Lemma \ref{union}. For this, we have to keep track of the existance of $0$- and $1$-vertices in the 2-neighbourhood-resolving sets that we compute. Since a $0$- or $1$-vertex with respect to a set $R$ becomes an $|R|$ or $(|R|-1)$-vertex when complementing, we also have to keep track of $|R|$- and $(|R|-1)$-vertices.\\
 For a cograph $G=(V,E)$ we define 16 types of minimum weight 2-neighbourhood-resolving sets $R_{a,b,c,d}$, $a,b,c,d \in \{0,1\}$.\\
 For 
 \begin{itemize}
  \item $a=1$ we compute a minimum weight 2-neighbourhood-resolving set $R$ for $G$ such that there is a 0-vertex in $G$ with respect to $R$ and for $a=0$ we compute a minimum weight 2-neighbourhood-resolving set for $G$ such that there is no 0-vertex in $G$ with respect to $R$.
  \item $b=1$ we compute a minimum weight 2-neighbourhood-resolving set $R$ for $G$ such that there is a 1-vertex in $G$ with respect to $R$ and for $b=0$ we compute a minimum weight 2-neighbourhood-resolving set for $G$ such that there is no 1-vertex in $G$ with respect to $R$.
  \item $c=1$ we compute a minimum weight 2-neighbourhood-resolving set $R$ for $G$ such that there is a $(|R|-1)$-vertex in $G$ with respect to $R$ and for $c=0$ we compute a minimum weight 2-neighbourhood-resolving set for $G$ such that there is no $(|R|-1)$-vertex in $G$ with respect to $R$.
  \item $d=1$ we compute a minimum weight 2-neighbourhood-resolving set $R$ for $G$ such that there is a $|R|$-vertex in $G$ with respect to $R$ and for $d=0$ we compute a minimum weight 2-neighbourhood-resolving set for $G$ such that there is no $|R|$-vertex in $G$ with respect to $R$.
 \end{itemize}

 Let $r_{a,b,c,d}$ be the weight of the corresponding minimum weight 2-neighbourhood-resolving sets $R_{a,b,c,d}$, i.e. the sum of the weights of all vertices in $R_{a,b,c,d}$. If there is no such 2-neighbourhood-resolving set for a certain $a,b,c,d$, we set $r_{a,b,c,d} = \infty$ and $R_{a,b,c,d}=undefined$.\\
 
 Now we will analyze the 16 2-neighbourhood-resolving sets more detailed and describe, how they can be computed efficiently bottom up along the cotree. First one should note that $r_{1,1,c,d}=\infty$, $\forall c,d$, and $R_{1,1,c,d}=undefined$, since it is not possible to have a $0$- and $1$-vertex with respect to $R$ in a 2-neighbourhood-resolving set (their symmetric difference would contain less than two resolving vertices), so it suffices to focus on the remaining 12 sets.
 
 When complementing a graph $G$, the role of a $0$-vertex and $|R|$-vertex with respect to $R$ and the role of a $1$-vertex and a $(|R|-1)$-vertex with respect to $R$ changes, that is $R_{a,b,c,d}$ for $G$ is $R_{d,c,b,a}$ for $\bar G$. When unifying two cographs $G_1$ and $G_2$ we distinguish between the follwing three cases:
 \begin{enumerate}
  \item $G_1$ and $G_2$ both consist of a single vertex
  \item $G_1$ consists of a single vertex and $G_2$ of at least two vertices
  \item $G_1$ and $G_2$ both consist of at least 2 vertices
 \end{enumerate} 

 We will describe now how to compute $R_{a,b,c,d}$ for the three cases.
 
 \begin{enumerate}
  \item Let $G_1=(\{v_1 \}, \emptyset)$ and $G_2=(\{v_2\}, \emptyset)$. Then there is exactly one valid 2-neighbourhood-resolving set for $G=G_1 \cup G_2$, namely $R=\{v_1,v_2\}$. In $G$ we have no $0$-vertex, two $1$- and two $(|R|-1)$-vertices and no $|R|$-vertex with respect to $R$. Therefore $R_{0,1,1,0}=\{v_1,v_2\}$, $r_{0,1,1,0}= c(v_1) + c(v_2)$ and all other sets are infeasible, that is $r_{a,b,c,d}=\infty$ and $R_{a,b,c,d}=undefined$ for $a\neq 0 \vee b \neq 1 \vee c\neq 1 \vee d\neq0$.
  
  \item Let $G_1= ( \{ v_1 \},\emptyset)$ and $G_2=(V_2,E_2)$ with $|V_2| \geq 2$. For some $a,b,c,d \in\{0,1\}$ let $R''_{a,b,c,d}$ be the minimum weight 2-neighbourhood-resolving sets for $G_2$ and $r''_{a,b,c,d}$ be their weights. Let $G=G_1 \cup G_2$. $r_{0,0,c,d} = \infty$ and $R_{0,0,c,d} = undefined$, because vertex $v_1$ is either a $0$-vertex (if it is not in the 2-neighbourhood-resolving set) or a $1$-vertex (if it is in the 2-neighbourhood-resolving set) with respect to $R_{0,0,c,d}$, $\forall c,d$. $r_{0,1,c,1} = \infty$ and $R_{0,1,c,1} = undefined$, because it is crucial to put $v_1$ in the 2-neighbourhood-resolving set, if there should be no $0$-vertex in $G$ with respect to $R_{0,1,c,1}$, $\forall c$. If $v_1$ is in the 2-neighbourhood-resolving set, it is not possible to have a vertex that is neighbour of all resolving vertices, because $v_1$ has no neighbours. For $R_{0,1,0,0}$ and $R_{0,1,1,0}$ we have to put $v_1$ in the 2-neighbourhood-resolving set, so that there is no $0$-vertex with respect to $R_{0,1,0,0}$ or $R_{0,1,1,0}$, what makes $v_1$ become a $1$-vertex in $G$ with respect to $R_{0,1,0,0}$ or $R_{0,1,1,0}$. We get $r_{0,1,0,0} = c(v_1) + \min \{r''_{0,0,0,0}, r''_{0,0,1,0}, r''_{0,1,0,0}, r''_{0,1,1,0}\}$ and thus  $R_{0,1,0,0}= \{v_1\} \cup R_m$, whereas $R_m$ is the set with smallest weight out of $\{R''_{0,0,0,0}, R''_{0,0,1,0}, R''_{0,1,0,0}, R''_{0,1,1,0}\}$. For $R_{0,1,1,0}$ there has to be an $|R_{0,1,1,0}|$-vertex in $G_2$ with respect to $R_{0,1,1,0}$, so we get $r_{0,1,1,0} = c(v_1) + \min \{r''_{0,0,0,1}, r''_{0,0,1,1}, r''_{0,1,0,1}, r''_{0,1,1,1}\}$ and thus $R_{0,1,1,0} = \{v_1\} \cup R_m$, whereas $R_m$ is the set with smallest weight out of $\{R''_{0,0,0,1}, R''_{0,0,1,1}, R''_{0,1,0,1}, R''_{0,1,1,1}\}$. For $R_{1,0,c,d}$ it is not possible to put $v_1$ in the 2-neighbourhood-resolving set, because it would become a $1$-vertex with respect to $R_{1,0,c,d}$, $\forall c,d$. Therefore we get $r_{1,0,0,0} = r''_{0,0,0,0}$ and thus $R_{1,0,0,0} = R''_{0,0,0,0}$, $r_{1,0,0,1} = r''_{0,0,0,1}$ and thus $R_{1,0,0,1} = R''_{0,0,0,1}$, $r_{1,0,1,0} = r''_{0,0,1,0}$ and thus $R_{1,0,1,0} = R''_{0,0,1,0}$, $r_{1,0,1,1} = r''_{0,0,1,1}$ and thus $R_{1,0,1,1} = R''_{0,0,1,1}$.
  
  \item Let $G_1=(V_1,E_1)$ and $G_2=(V_2,E_2)$ with $|V_1| \geq 2$ and $|V_2| \geq 2$ and  $G=G_1 \cup G_2$. For some $a,b,c,d \in \{0,1\}$ let $R'_{a,b,c,d}$ be the minimum weight 2-neighbourhood-resolving sets for $G_1$ and $R''_{a,b,c,d}$ be the minimum weight 2-neighbourhood-resolving sets for $G_2$ and $r'_{a,b,c,d}$ and $r''_{a,b,c,d}$ be their weights. $r_{a,b,c,1}= \infty$ and $r_{a,b,1,d}=\infty$ and thus $R_{a,b,c,1}= undefined$ and $R_{a,b,1,d}=undefined$, $\forall a,b,c,d$, because $G_1$ and $G_2$ contain at least two resolving vertices in every 2-neighbourhood-resolving set. Therefore it is not possible to have a vertex that is neighbour of all or of all except one of them. The three remaining sets are $R_{0,0,0,0}, R_{0,1,0,0}, R_{1,0,0,0}$. 
  We get $r_{0,0,0,0} = \min \{r'_{0,0,c,d} | c,d \in \{0,1\}\} + \min \{r''_{0,0,c,d} | c,d \in \{0,1\}\}$ and thus $R_{0,0,0,0} = R'_m \cup R''_m$, whereas $R'_m$ is the set with smallest weight out of $\{R'_{0,0,c,d} | c,d \in \{0,1\}\}$ and $R''_m$ is the set with smallest weight out of $\{R''_{0,0,c,d} | c,d \in \{0,1\}\}$. 
  We get $r_{0,1,0,0} = \min \{r'_{0,0,c,d} + r''_{0,1,c',d'}, r'_{0,1,c,d} + r''_{0,0,c',d'}, r'_{0,1,c,d} + r''_{0,1,c',d'} | c,d,\linebreak  c',d' \in \{0,1\} \}$ and thus $R_{0,1,0,0} = \min \{ R'_{m_0} \cup  R''_{m_1}, R'_{m_1} \cup R''_{m_0}, R'_{m_1} \cup R''_{m_1} \}$, whereas $R'_{m_0}$ is the set with smallest weight out of $\{R'_{0,0,c,d} | c,d \in \{0,1\}\}$, $R'_{m_1}$ is the set with smallest weight out of $\{R'_{0,1,c,d} | c,d \in \{0,1\}\}$, $R''_{m_0}$ is the set with smallest weight out of $\{R''_{0,0,c,d} | c,d \in \{0,1\}\}$ and $R''_{m_1}$ is the set with smallest weight out of $\{R''_{0,1,c,d} | c,d \in \{0,1\}\}$. 
  We get $r_{1,0,0,0} = \min \{r'_{1,0,c,d} + r''_{0,0,c',d'}, r'_{0,0,c,d} + r''_{1,0,c',d'} | c,d, c',d' \in \{0,1\} \}$ and thus $R_{1,0,0,0} = \min \{R'_{m_1} \cup R''_{m_0}, R'_{m_0} \cup R''_{m_1}\}$, whereas $R'_{m_0}$ is the set with smallest weight out of $\{R'_{0,0,c,d} | c,d \in \{0,1\}\}$, $R'_{m_1}$ is the set with smallest weight out of $\{R'_{1,0,c,d} | c,d \in \{0,1\}\}$, $R''_{m_0}$ is the set with smallest weight out of $\{R''_{0,0,c,d} | c,\linebreak d \in \{0,1\}\}$ and $R''_{m_1}$ is the set with smallest weight out of $\{R''_{1,0,c,d} | c,d \in \{0,1\}\}$.
 \end{enumerate}
 
 For every node of the cotree $T$ the computation of the 12 minimum weight 2-neighbourhood-resolving sets for the corresponding subgraph of $G$ can be done in a constant number of steps. Since $T$ has $\mathcal{O}(n)$ nodes, the overall runtime of our algorithm is linear to the size of the cotree.
 
\end{proof}

\section{Conclusion}
We showed that the weighted fault-tolerant metric dimension problem can be solved in linear time on cographs. Our algorithm computes the costs of a fault-tolerant resolving set with minimum weight as well as the set itself.

The complexity of computing the (weighted) fault-tolerant metric dimension is still unknown even for graph classes like wheels and sun graphs. This is something that we will investigate in further work.

%
%

\bibliographystyle{splncs04}
\bibliography{literature}

\end{document}